\documentclass[final,nomarks]{dmtcs-episciences}
\usepackage[utf8]{inputenc}
\usepackage{subfigure}
\usepackage{amssymb,amsmath,amsthm,textcomp}
\usepackage[round]{natbib}

\newtheorem{theorem}{Theorem}[section]
\newtheorem*{theorem*}{Theorem}

\newtheorem{lemma}[theorem]{Lemma}
\newtheorem{remark}{Remark}

\newtheorem{claim}{Claim}

\newtheorem*{conjecture*}{Conjecture}

\author{Arpitha P. Bharathi
  \and Minati De\thanks{Supported by DST-INSPIRE Faculty Grant DST-IFA14-ENG-75.}
  \and Abhiruk Lahiri}
\title{Circular Separation Dimension of a Subclass of Planar Graphs}

\affiliation{Department of Computer Science and Automation, Indian Institute of Science, Bangalore, India}


\received{2017-1-20}
\accepted{2017-9-21}
\begin{document}
\publicationdetails{19}{2017}{3}{8}{2661}
\maketitle
\begin{abstract}
A pair of non-adjacent edges is said to be \emph{separated} in a circular ordering of vertices, if the endpoints of the two edges do not alternate in the ordering. The \emph{circular separation dimension} of a graph $G$, denoted by $\pi^\circ(G)$, is the minimum number of circular orderings of the vertices of $G$ such that every pair of non-adjacent edges is separated in at least one of the circular orderings. This notion is introduced by Loeb and West in their recent paper. In this article, we consider two subclasses of planar graphs, namely $2$-outerplanar graphs and series-parallel graphs. A $2$-outerplanar graph has a planar embedding such that the subgraph obtained by removal of the vertices of the exterior face is outerplanar. We prove that if $G$ is $2$-outerplanar then $\pi^\circ(G) = 2$. We also prove that if $G$ is a series-parallel graph then $\pi^\circ(G) \leq 2$.
\end{abstract}

\section{Introduction}
\noindent
\cite{BasavarajuCGMR16} introduced the notion of separation dimension of a graph. \emph{Separation dimension} $\pi(G)$ of a graph $G=(V,E)$, denotes the minimum number of linear orderings of the vertices in $V$ required to ``separate'' all non-adjacent edges in $E$. Here, a pair $(e,f)$ of edges is separated in an ordering if both vertices of $e$ appear after/before both vertices of $f$. 

Very recently, \cite{LoebWest16} introduced a similar terminology circular separation dimension. The \emph{circular separation dimension} of a graph $G=(V,E)$, denoted by $\pi^\circ(G)$, is the minimum number of circular orderings of $V$ needed to separate all non-adjacent edges in $E$. Here, a pair  of edges is separated in an ordering    if the  endpoints of the two edges do not alternate in that ordering. As a pair of edges separated  in a linear ordering are also separated  in  the  corresponding circular ordering, we know that $\pi^\circ(G)\leq \pi(G)$. Another variant of the dimensional problem that has been recently studied is the \textit{induced separation dimension} by \cite{Ziedan}.

Addition of edges in a graph does not  decrease both the parameters $\pi$ and $\pi^\circ$; this property is referred to as \textit{monotonicity}. Thus, for a graph with $n$ vertices, these parameters achieve the maximum value when the graph is a complete graph $K_n$. \cite{BasavarajuCGMR16} showed that  in general for any graph with $n$ vertices, $ \log_2(\lfloor\frac{1}{2} \omega(G)\rfloor) \leq \pi(G)\leq 4 \log_{\frac{3}{2}} n$. \cite{LoebWest16} proved that  $\pi^\circ(G)> \log_2 \log_3(\omega(G)-1)$. Here $\omega(G)=\max\{t|K_t\subseteq G\}$.   

The ranges of both the parameters were studied for  special  graph classes as well. \cite{AlonBCMR15} showed that  $\pi(G)\leq  2^{9{log^{\star}}\!d} d$ for a graph $G$ whose degree is bounded by   $d$. They also proved that for almost all $d$ regular graphs,  $\pi(G)$ is at least $\lceil d/2 \rceil$. For a complete bipartite graph $K_{m,n}$, $\pi(K_{m,n})\geq \log_2\min\{m,n\}$~\cite{BasavarajuCGMR16}, and $\pi^\circ(K_{m,n})=2$~\cite{LoebWest16}.

If $G$ is a planar graph, $\pi(G) \leq 3$  \cite{BasavarajuCGMR16}. Clearly from the previous observation \textit{i.e.} $\pi^\circ(G)\leq \pi(G)$, we have that $\pi^\circ(G) \leq 3$ when $G$ is planar.  On the other hand, it is  interesting to note that $\pi^\circ(G)=1$ if  and only if $G$ is outerplanar. This result is again by \cite{LoebWest16}. 
Thus,  it follows easily for planar graphs that are not outerplanar, that $\pi^\circ (G)$ is either 2 or 3. In this context, we  conjecture  the following.

\begin{conjecture*}
The circular separation dimension of a  planar graph is at most two.
\end{conjecture*}

Note that even for $K_4$, the smallest  2-outerplanar graph, we have $\pi^\circ(K_4)>1$. On the other hand, from the result of~\cite{BasavarajuCGMR16}, we know that  $\pi(K_4)=3$.  But two circular permutations are enough to separate the edges of $K_4$ (if $\{a,b,c,d\}= V(K_4)$ then the permutations $(a,b,c,d)$ and $(a,c,b,d)$ separate all pairs of non-adjacent edges).

A planar graph is said to be $k$-outerplanar, $k \geq 2$, if it has a planar embedding such that by removing the vertices on the unbounded face we obtain a $(k-1)$-outerplanar graph.  It is a well known fact that every planar graph is $k$-outerplanar for some integer $k$ (typically, much smaller than $n$)~\cite{Bienstock}. Thus,  it  is  natural to  investigate  the circular separation dimension of $2$-outerplanar  graphs, and to see whether one can generalize the result for $k$-outerplanar graphs. In Section~\ref{sec: 2-out}, we  prove that the circular separation dimension of a $2$-outerplanar  graph is  two.

It is interesting to note that for series-parallel graphs, which is $k$-outerplanar for some positive integer $k$, two circular permutations are sufficient to separate all non-adjacent pairs of edges. We discuss this result in Section~\ref{sec: series}.

\subsection{Preliminaries}
A graph is \emph{outerplanar} if it has a planar embedding such that all vertices are on the outer face. A graph $G$ is \emph{$2$-outerplanar} if $G$ has a planar embedding such that the subgraph obtained by removal of the vertices of the exterior face is outerplanar. In this paper, we consider $G$ with its $2$-outerplanar embedding.\\

\noindent A graph is a \emph{series-parallel} graph if it can be turned into a $K_2$ by a sequence of the following operations.
\begin{itemize}
\item []Replacement of a pair of parallel edges with a single edge that connects their common endpoints. This operations is called a \emph{parallel operation}.
\item []Replacement of a pair of edges incident to a vertex of degree $2$ other than two distinguished vertices, called source and sink vertex, with a single edge. This operation is called a \emph{series operation}.
\end{itemize}

\noindent A \emph{cut vertex} of a connected graph is a vertex whose deletion disconnects the graph. 
A connected graph is \textit{biconnected} if  it requires  deletion of  at least two vertices to disconnect the graph.
We consider a \emph{block} of a graph to be the vertex set of a maximal biconnected subgraph. Other terminologies, which are not defined here, can be found in~\cite{Diestel}. 

\subsection{Definitions and Notations}
Let  $G=(V,E)$ be  a graph with vertex set $V$ and edge set $E$. Let $\sigma : V \rightarrow \{1,2,\dots, |V|\}$ be a permutation of elements of the vertex set $V$.  A permutation  $\sigma = (v_1,v_2,\dots,v_k)$ of $V$ implies that $\sigma(v_i)<\sigma(v_j), \; 1\leq i < j \leq k$.  
A  \textit{sub-permutation}  $\sigma'$ of  $\sigma$ restricted to  $V'\subseteq V$ is defined as $\sigma':V'\rightarrow\{1,2,\dots,|V'|\}$ such that $\sigma'(v_i)<\sigma'(v_j)$ if and only if  $\sigma(v_i)<\sigma(v_j)$ for all $v_i,v_j\in V'$. We define $reversal(\sigma)= (v_k,v_{k-1},\dots,v_2,v_1)$. If $\alpha=(a_1,a_2,\dots,a_n)$ and $\beta=(b_1,b_2,\dots,b_m)$ are two permutations, then the permutation $(\alpha,\beta)$ is defined as $(a_1,a_2,\dots,a_n,b_1,b_2,\dots,b_m)$ where $(\alpha,\beta)(a_i)=\alpha(a_i)$ and $(\alpha,\beta)(b_i)=\beta(b_i)+n$. 
A pair of non-adjacent edges in  $G$ is \emph{separated} by a circular ordering or circular permutation of $V$ if the endpoints of the two edges do not alternate. 
If two non-adjacent edges are not separated, then they are said to \textit{cross} each other in the permutation. If a pair of edges cross each other in the first permutation and are separated in the second, then the second permutation is said to \textit{resolve} these pair of edges. A family $\sigma$ of circular permutations of $V$ is called \textit{pairwise suitable} for $G$ if, for every pair of non-adjacent edges in $G$, there exists a permutation in $\sigma$ in which the edges are separated. The \emph{circular separation dimension} of a graph $G$, denoted by $\pi^\circ(G)$, is the minimum cardinality of such a family. We assume that every graph under consideration is connected, as $\pi^\circ(G) = \max\{\pi^\circ(H)\:|\:H$ is a component of $G \}$. We use $N(v)$  to denote the set of all neighbours of $v$.

\section{2-Outerplanar graphs}
\label{sec: 2-out}
In this section, we prove the following main theorem.

\begin{theorem*}
The circular separation dimension of a $2$-outerplanar graph is exactly two.
\end{theorem*}

We use the following result on outerplanar graphs which was proved by Loeb and West.
\begin{lemma}[\cite{LoebWest16}]
\label{lemma:outerplanar_perm}
If $G$ is a maximal outerplanar graph, then any circular ordering of vertices following the outer face produces a circular permutation that separates all pairs of non-adjacent  edges of $G$.
\end{lemma}

\par Let $G=(V,E)$ be a maximal 2-outerplanar graph in the sense that every face is a triangle except the outer face. We fix a planar embedding of $G$ and  our proofs are based on such an embedding. Let $V_2$ be the set of vertices appearing on the exterior face of $G$, and $V_1=V\setminus V_2$. Let $n_i = |V_i|$ for $i\in \{1,2\}$. Let $E_i$ be the set of edges in $[V_i]$. Clearly, each $[V_i]$ is a 1-outerplanar graph. Let $E_{12} = E(G) \setminus (E_1\cup E_2)$ (this is exactly the set of edges with one end point in $V_1$ and the other end point in $V_2$). Thus, we have a partition of the edge set of $G$ as $E(G)=E_1 \cup E_2 \cup E_{12}$. To verify if a family of circular permutations is pairwise suitable, we check if there can be any \textit{(i)} $E_1-E_1$ crossing \textit{(ii)} $E_2-E_2$ crossing \textit{(iii)} $E_1-E_2$ crossing \textit{(iv)} $E_1-E_{12}$ crossing  \textit{(v)} $E_{12}-E_2$ crossing \textit{(vi)} $E_{12}-E_{12}$ crossing in both permutations.
Let a vertex of $V_i$ be denoted as $s^i$.

Note that if $[V_2]$ contains a cut vertex, then by adding edges on the outer face between neighbours of the cut vertex (retaining planarity), one can  make $[V_2]$ biconnected (see Figure \ref{fig:OuterLayerContainsCutVertex}). The addition of edges does not decrease the value of $\pi^\circ$ (because of its monotonicity property).  Thus, we only consider the case when the graph $[V_2]$ is biconnected.  \\

\begin{figure}[t]
\begin{center}
    \subfigure[$\lbrack V_2 \rbrack$ is not biconnected.\label{fig:OuterLayerCutVertex}]{\includegraphics[width=3.5cm]{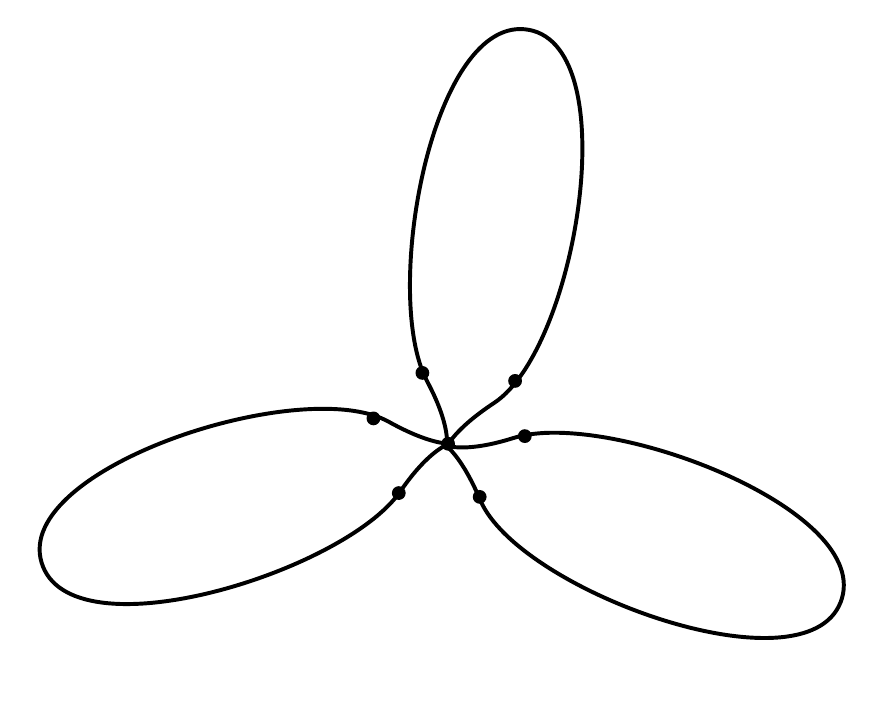}}
    \hspace{3cm}
    \subfigure[$\lbrack V_2 \rbrack$ is biconnected.\label{fig:OuterLayerCutVertexWithEdges}]{\includegraphics[width=3.5cm]{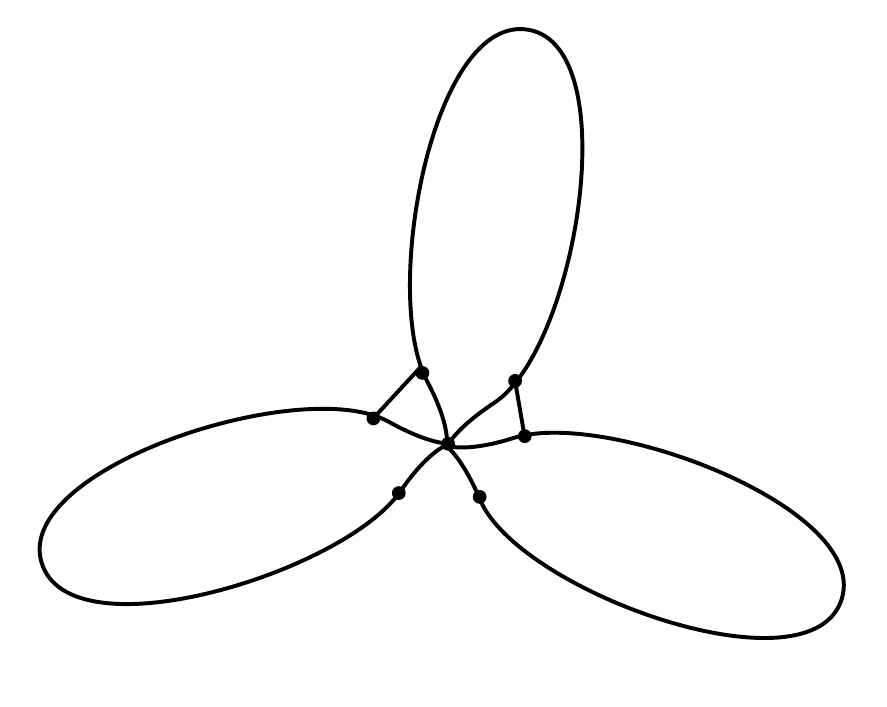}}
    \caption{Suitable addition of edges when $[V_2]$ contains a cut vertex.}
    \label{fig:OuterLayerContainsCutVertex}
\end{center}
\end{figure}

\par  We first introduce  a technique  which is often used to resolve $E_1-E_{12}$ crossing called the ``arc-removal'' technique.\\

\begin{figure}[b]
\begin{center}
    \subfigure[Before arc-removal technique.\label{fig:Arc_removal_technique}]{\includegraphics[width=6cm]{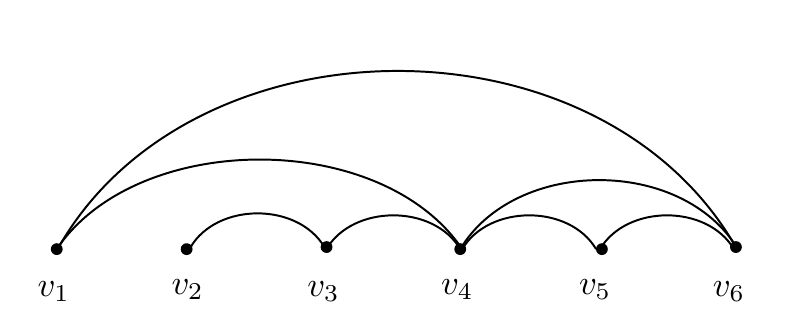}}
    \hspace{2cm}
    \subfigure[After arc-removal technique fixing the right end.\label{fig:After_arc_removal_technique}]{\includegraphics[width=6cm]{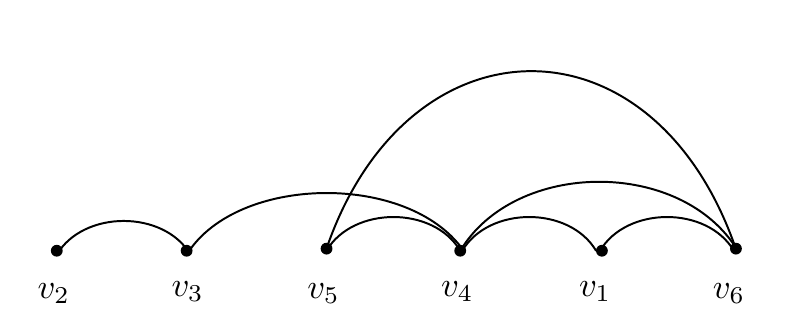}}
    \caption{Resolving $E_{1}-E_{12}$ crossing.}
    \label{fig:arcrem}
\end{center}
\end{figure}

\par Let  $G$ be a maximal  2-outerplanar graph  such that $[V_1]$ is  biconnected. Let $\sigma'$ be a permutation of $V_1$ according to Lemma \ref{lemma:outerplanar_perm}. Similarly, let $(s_1^2,s_2^2,\dots,s_{n_2}^2)$ be a permutation of $V_2$. Let $\sigma_1 = (s_1^2,\sigma',s_2^2,s_3^2,\dots,s_{n_2}^2)$. Suppose $e=(x,y) \in E_1$ where $\sigma'(x)<\sigma'(y)$. Let $f\in E_{12}$ cross  $e$ in $\sigma_1$. Observe that each $v\in V_1$ such that $\sigma'(x)<\sigma'(v)<\sigma'(y)$ is a ``candidate" for an endpoint of  $f$. 
To separate $e$ and $f$,  we construct a permutation $\sigma''$ where each $v$ appears either before or after  both  $x$ and $y$. We need to do this for all of the edges in $E_1$. Here we use the following fact: As all the non-adjacent pairs of edges  in $E_1$ are separated in $\sigma'$, all of the edges in $E_1$ form a ``well-parenthesis" structure in the ordering $\sigma'$ where each parenthesis represents an edge (see Figure \ref{fig:Arc_removal_technique}).  So,  we  can ``unfold" the ``parenthesis" one after another (see Figure \ref{fig:After_arc_removal_technique}).\\

The  arc-removal technique  takes as input the permutation $\sigma'$ and a parameter $p\in \{r,l\}$  which  decides whether we need  to  maintain  the position of the right most vertex or the  left most vertex  of $\sigma'$. The arc-removal technique fixing the right end is explained as an algorithm below.\\

\par \noindent \textbf{Input:} Edge separating permutation $\sigma' = (v_1,v_2,\dots,v_{n_1})$ of vertices  in $[V_1]$  and a parameter $p\in \{l,r\}$.\\
\noindent \textbf{Output:} A permutation $\sigma''$ of $\{v_1,v_2,\dots,v_{n_1}\}$.\\
If  $p=r$, then we perform the following.\\
\noindent Initialization: Set $i = n_1$ and set $k = n_1$. Mark every vertex as \texttt{not selected}. Let $\sigma''$ be an empty ordering.
\begin{itemize}
\item[1.] If $v_i$ is \texttt{selected} go to step 4. If $v_i$ is \texttt{not selected} check if $v_i$ has already appeared in $\sigma''$. If not, set $\sigma''(v_i) = k$. Set $k = k-1$.
\item[2.] Choose minimum $j$ for all $v_j \in N(v_i)$ where $v_j$ does not appear in $\sigma''$. Set $\sigma''(v_j) = k$ and  $k= k-1$.
\item[3.] Repeat step 2 till all the vertices in $N(v_i)$ appear in $\sigma''$. Mark $v_i$ as \texttt{selected}.
\item[4.] Set $i = i-1$. Go to step 1 if $k>0$.
\end{itemize}
If $p=l$,  then we perform a similar technique as described above. The only differences are: \textit{(i)} In the initialization, we set both $i,k=1$ in order to fix the left end; \textit{(ii)} Instead of decrementing, we increment these variables in steps 1, 2 and 4; \textit{(iii)} In step 2, we choose the maximum $j$ among neighbours \texttt{not selected}; \textit{(iv)} In step 4, we go to step 1 if $k<n_1$.
\begin{claim}\label{claim:arcremovalworks}
Let  $\sigma'$ be a permutation of $V_1$ that separates all non-adjacent pairs of edges in $[V_1]$, and $\sigma_1 = (s_1^2,\sigma',s_2^2,s_3^2,\dots,s_{n_2}^2)$. Then the permutation $\sigma'' = (arc\mbox{-}removal(\sigma',p),s_1^2,s_2^2,...,s_{n_2}^2)$ resolves every pair of crossing edges $e,f$, where $e \in E_{12},f\in E_1$ and $e,f$ are crossing in $\sigma_1$, and $p \in \{l,r\}$. 
\end{claim}
\begin{proof}
Following our notation, suppose $e = (s^1_i, s^2_j)$ and $f = (s^1_l, s^1_m)$ where without loss of generality, we can assume that $1\leq l<i<m \leq n_1, 1\leq j \leq n_2$. Suppose $p=r$. Step 2 of the algorithm ensures that $s_i^1$ appears before $s^1_l$ and $s^1_m$. Thus in $\sigma''$, $s^1_l$ and $s^1_m$ appear in between $s^1_i$ and $s_j^2$  and therefore $e,f$ are separated. It should be noted that $s_i^1$ cannot be listed in $\sigma''$ before $s_m^1$ is listed, as this would mean there is an edge $g=(s_i^1,s_h^1)$ such that $m<h$, a contradiction as $f,g$ would cross in $\sigma'$. When $p=l$, $s^1_l$ and $s^1_m$ appear before $s^1_i$ and $s_j^2$ in $\sigma''$.
\end{proof}

\begin{claim}\label{claim:reversalworks}
Let  $\sigma'$ be a permutation of $V_1$ that separates all non-adjacent pairs of edges in $[V_1]$, and $\sigma_1 = (s_1^2,\sigma',s_2^2,s_3^2,\dots,s_{n_2}^2)$. Then the permutation $\sigma'' = (reversal(arc\mbox{-}removal(\sigma',p)),$ $s_1^2,s_2^2,...,s_{n_2}^2)$ resolves every pair of crossing edges $e,f$, where $e \in E_{12},f\in E_1$ and $e,f$ are crossing in $\sigma_1$, and $p \in \{l,r\}$.
\end{claim}
\begin{proof}
Consider edges $e,f$ as in Claim \ref{claim:arcremovalworks}. Suppose $p=r$. Then $(arc\mbox{-}removal(\sigma',p)),$ $s_1^2,s_2^2,...,s_{n_2}^2)$ resolves $e,f$ crossing as $s^1_l$ and $s^1_m$ appear in between $s^1_i$ and $s_j^2$ due to Claim~\ref{claim:arcremovalworks}. The reversal would now imply that $s^1_l$ and $s^1_m$ appear before $s^1_i$ and $s_j^2$ in $\sigma''$. If $p=l$, then $s^1_l$ and $s^1_m$ appear in between $s^1_i$ and $s_j^2$ in $\sigma''$.
\end{proof}

We now prove two lemmas that help us to prove the main theorem of this section.
\begin{lemma}
\label{lemma:biconnected}
Let $G$ be a maximal $2$-outerplanar graph. If the graph induced on $V_1$ is biconnected, then $\pi^\circ(G) = 2$.
\end{lemma}
\begin{proof}
Let $s_1^2$ be a vertex in $V_2$ that has a neighbour $x$ in $V_1$ (such a vertex exists, else $G$ is disconnected). We name $s_1^2$ as the \textit{start vertex} of $G$.\\ 
\textit{Choosing $s_1^1$:} Upon moving counter-clockwise from $x$ on the outer boundary of $[V_1]$, let $e$ (see Figure \ref{fig:2-layers}) be the first edge of $E_{12}$ that is seen after $s_1^2x$ that does not have an endpoint in $s_1^2$.  Such an edge exists, else all vertices of $V_1$ are adjacent to $s_1^2$ and no other vertex of $V_2$, which contradicts the maximality of $G$ (note that $n_2\geq 3$ as $V_2$ are the vertices of the outer face, see Figure \ref{fig:2-layers}). 
Let the endpoint of $e$ in $V_1$ be $s_1^1$ (it could very well be possible that $x=s_1^1$). Note that $s_1^1s_1^2$ is an edge due to the maximality of $G$.\\
\noindent\textit{Defining $e'$:} Upon moving clockwise from $s_1^1$ on the outer face of $[V_1]$, let $e'$ be the first edge of $E_{12}$ that is seen after $s_1^1s_1^2$. Then the edge $e'$ either has $s_1^2$ or $s_1^1$ as an endpoint due to the maximality of $G$.\\

Let $s_1^1,s_2^1,...,s_{n_1}^1$ be the vertices of $[V_1]$ traversed in a clockwise order on its boundary (it is useful to know that a biconnected outerplanar graph is Hamiltonian~\cite{Syslo}). Similarly, let $s_1^2,s_2^2,...,s_{n_2}^2$ be the vertices of $[V_2]$ traversed in an anti-clockwise order on its boundary. Depending on whether $e'$ is incident  with $s_1^1$ or $s_1^2$, consider  the first circular permutation $\sigma_1= (s_1^2,s_2^1,s_3^1,s_4^1,...,s_{n_1}^1,s_1^1,s_2^2,...,s_{n_2}^2)$ or $\sigma_1= (s_1^2,s_1^1,s_2^1,s_3^1,...,s_{n_1}^1,s_2^2,...,s_{n_2}^2)$. We now observe which pair of edges from $E_1,E_2,E_{12}$ can or cannot cross in both the cases of $\sigma_1$.\\

\par \noindent \textit{$E_1-E_1$ crossing, $E_2-E_2$ crossing, $E_1-E_2$ crossing:} It can be observed that in $\sigma_1$ there is no pairwise crossing of edges that are both taken from $E_i$, since $[V_i]$ is 1-outerplanar and thus applying Lemma \ref{lemma:outerplanar_perm}. Clearly, no edge of $E_1$ crosses any edge of $E_2$ in $\sigma_1$.\\

\begin{figure}[t]
\begin{center}
    \subfigure[Existence of $e$ and $s_1^1$.\label{fig:2-layers}]{\includegraphics[width=4cm]{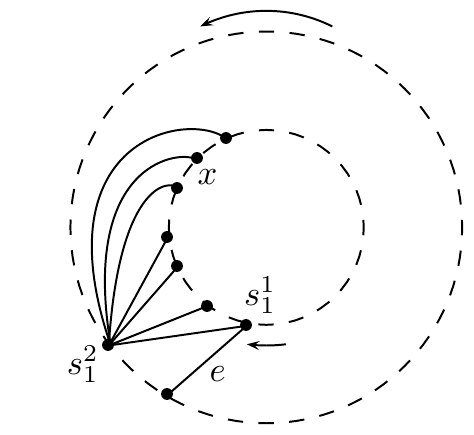}}
    \hspace{3cm}
    \subfigure[Edges of $B$ do not cross each other in $\sigma_1$.\label{fig:greengreen}]{\includegraphics[width=4cm]{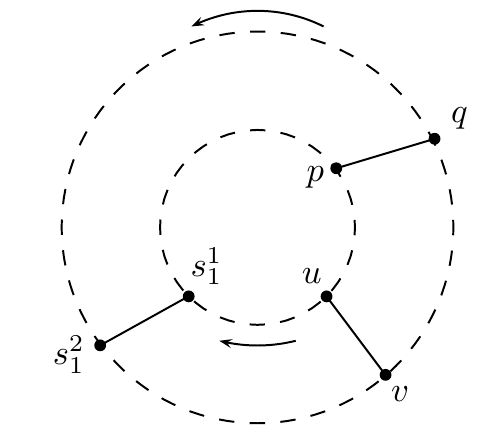}}
    \caption{The edges between the layers $V_1$ and $V_2$.}
    \label{fig:twolayers}
\end{center}
\end{figure}

Now, we produce  two pairwise suitable permutations based on two cases of $e'$ as follows:\\
\par \noindent \underline{Case 1: $e'$ has $s_1^2$ as an endpoint} (see Figure \ref{fig:basecasecase1}).\\
\indent Let $\alpha=(s_2^1,s_3^1,s_4^1,\dots,s_{n_1}^1,s_1^1)$. We define $\sigma_1= (s_1^2,\alpha,s_2^2,s_3^2,\dots,s_{n_2}^2)$ and $\sigma_2= (arc\mbox{-}removal(\alpha,r)$ $,s_1^2$ $,s_2^2,\dots,s_{n_2}^2)$. We now justify that $\sigma_2$ is a permutation that resolves crossing edges in $\sigma_1$.\\ 
\textit{$E_1-E_{12}$ crossing:} The crossings between edges of $E_1$ and $E_{12}$ are resolved by performing the arc removal technique on $\alpha$ (see Claim \ref{claim:arcremovalworks}). \\
\textit{$E_{12}-E_2$ crossing:} It is not possible that an edge from $E_{12}$ crosses an edge from $E_2$ in $\sigma_2$ except for $s_1^2s_{n_2}^2$, due to the planarity of $G$. The crossing of edges with $s_1^2s_{n_2}^2$ is already resolved in $\sigma_1$.\\ 
\textit{$E_{12}-E_{12}$ crossing:} It is not possible for any edge to cross $s_1^1s_1^2$ as $s_1^1$ and $s_1^2$ appear consecutively in both $\sigma_1$ and $\sigma_2$. Let $A$ denote the set of edges from $E_{12}$ that are incident on $s_1^2$ except for the edge $s_1^1s_1^2$. Then there exists an integer $k$ such that the endpoints of the edges of $A$ in $V_1$ are $A'=\{s_2^1,s_{3}^1,\dots,s_{k-2}^1,s_{k-1}^1,s_{k}^1,\}$ (it might very well be possible that $k=2$).\\ 
It is not possible for an edge from $A$ to cross an edge from $B= E_{12} \setminus (A\cup \{s_1^1s_1^2\})$ in $\sigma_1$ as all end points of edges in $A$ appear together.  We observe that no pair of edges from $B$ cross each other in $\sigma_1$. This is because of the following: suppose $pq,uv \in B$ and $p,u \in V_1$ and $q,v\in V_2$. Without loss of generality, we can assume $\sigma_1(p)<\sigma_1(u)$ (see Figure \ref{fig:greengreen}). 
This implies that $\sigma_1(v)<\sigma_1(q)$, else $pq$ and $uv$ cross in the planar embedding of $G$. Since neither $q$ nor $v$ can be $s_1^2$, the vertices $p,u$ appear before $v,q$ in $\sigma_1$. Hence we have $\sigma_1(p)<\sigma_1(u)<\sigma_1(v)<\sigma_1(q)$, thus separating $pq$ and $uv$ in $\sigma_1$.\\

\begin{figure}[t]
\begin{center}
    \subfigure[$e'$ has $s_1^2$ as an endpoint.\label{fig:basecasecase1}]{\includegraphics[width=6cm]{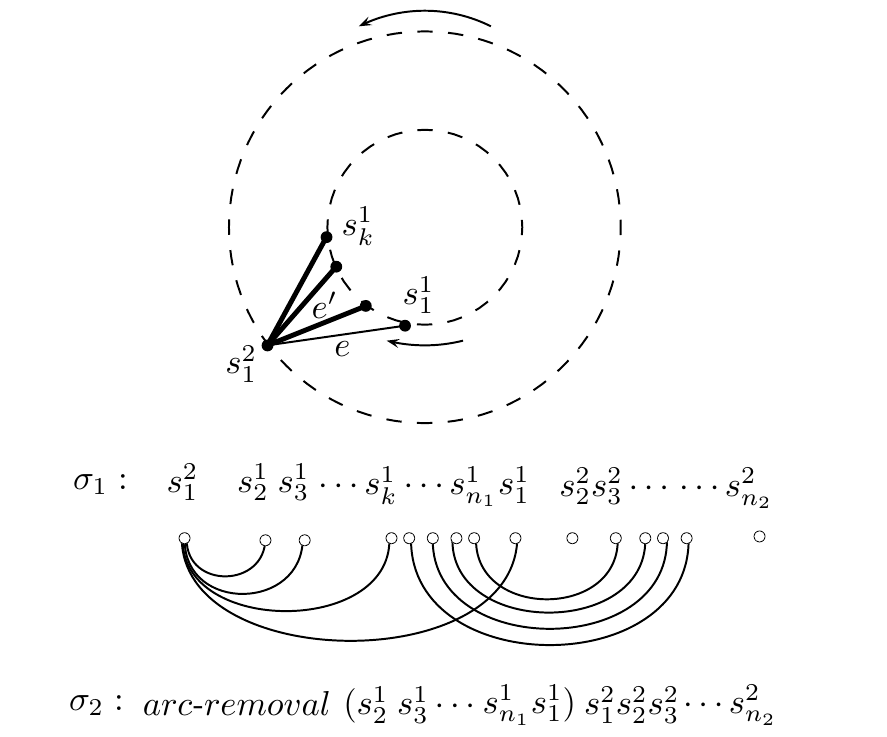}}
    \hspace{2cm}
    \subfigure[$e'$ has $s_1^1$ as an endpoint.\label{fig:basecasecase2}]{\includegraphics[width=6cm]{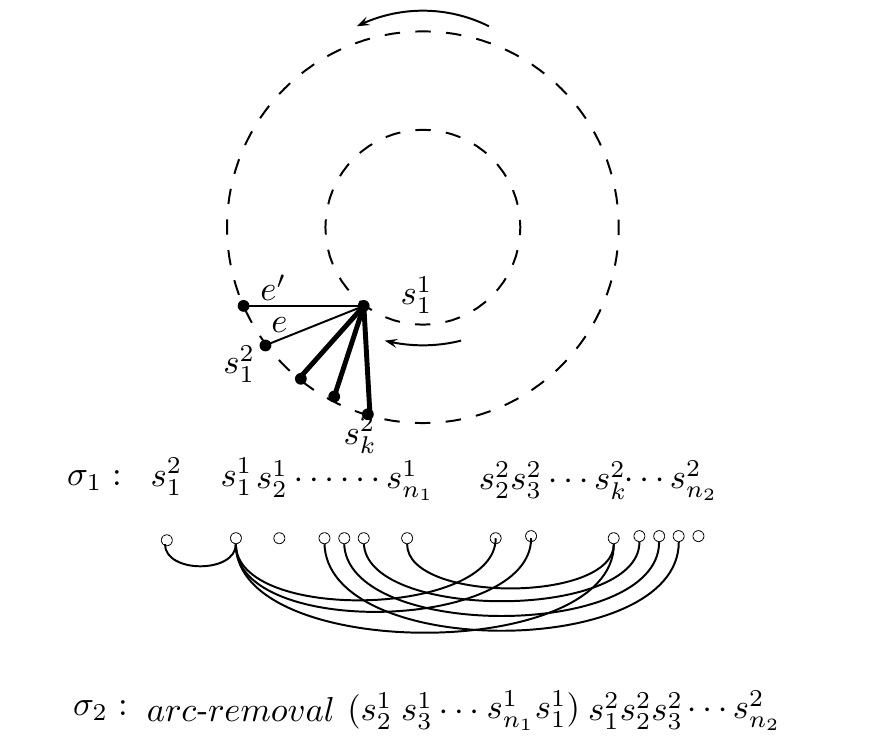}}
    \caption{Pairwise suitable family when $[V_1]$ is biconnected (edges of set $A$ are marked in bold, $arc\mbox{-}rem$ and $rev$ denote arc-removal and reversal respectively).}
    \label{fig:basecase}
\end{center}
\end{figure}

\noindent\underline{Case 2: $e'$ has $s_1^1$ as an endpoint} (see Figure \ref{fig:basecasecase2}).\\
\indent For this case, let $\alpha=(s_1^1,s_2^1,s_3^1,\dots,s_{n_1}^1)$. We define $\sigma_1= (s_1^2,\alpha,s_2^2,s_3^2,\dots,s_{n_2}^2)$.
Define $\sigma_2= (reversal$ $(arc\mbox{-}removal(\alpha,l)),s_1^2$ $,s_2^2,\dots,s_{n_2}^2)$. We now justify that $\sigma_2$ is a permutation that resolves all of the crossing edges in $\sigma_1$.\\
\textit{$E_{12}-E_2$ crossing:} The same reason as Case 1 holds.\\
\noindent \textit{$E_{12}-E_{12}$ crossing:} On moving counter-clockwise on the boundary of $[V_2]$ starting from $s_1^2$, let $s_k^2$ be the first vertex that has a neighbour in $V_1$ other than $s_1^1$ (it might very well be possible that $k=2$). Due to maximality of $G$, $s_1^1$ is adjacent to $s_k^2$. Let $A$ denote the set of edges between $s_1^1$ and $s_2^2,s_3^2,\dots,s_k^2$. Just as in Case 1, no pair of edges from $B = E_{12} \setminus (A\cup \{s_1^1s_1^2\})$ cross each other in $\sigma_1$. Thus the only pair of edges from $E_{12}$ that can cross is an edge from $A$ and an edge from $B$. These crossings are resolved by shifting $s_1^1$ to just before $s_1^2$ as given in $\sigma_2$. We retain this position of $s_1^1$ to prevent these edges from crossing again.\\
\textit{$E_{1}-E_{12}$ crossing:} These crossings are resolved by performing the arc removal technique on $\alpha=(s_1^1,s_2^1,\dots,$ $s_{n_1}^1)$ by fixing the left end and finally reversing this permutation.

\end{proof}
\begin{lemma}\label{lemma:connected}
Let $G$ be a maximal $2$-outerplanar graph. If the graph induced on $V_1$ is connected, then $\pi^\circ(G) = 2$.
\end{lemma}

\begin{proof}
We consider that case where $[V_1]$ is not biconnected. We follow a technique of listing the vertices that is almost identical to the one in Lemma \ref{lemma:biconnected}. Let $s_1^2$ be a vertex in $V_2$ that has a neighbour $x$ in $V_1$.\\
\noindent \textit{Choosing $s_1^1$, and defining $e'$:} This is done  as in Lemma \ref{lemma:biconnected}.\\
\textit{Labelling vertices and blocks:} Consider a clockwise walk along the outer face of $[V_2]$ starting from $s_1^2$. As we move along the outer face, we see one set of endpoints of edges in $E_{12}$. Consider the order in which the corresponding endpoints in $V_1$ are seen, which would also be clockwise along the outer face of $[V_1]$. We label the vertices of $V_1$ starting from $s^1_1$ in this order. We preserve the first occurrence of cut vertices and ignore repeated instances. Hence, in order to obtain the first permutation $\sigma_1$, we traverse clockwise along the boundary of the outer face of $[V_1]$ and in a counter-clockwise manner along the boundary of the outer face of $[V_2]$ (as in the previous lemma). 
We now aim to construct the second permutation through repeated contraction of blocks.\\

\begin{figure}[t]
\centering
\includegraphics[width=0.5\textwidth]{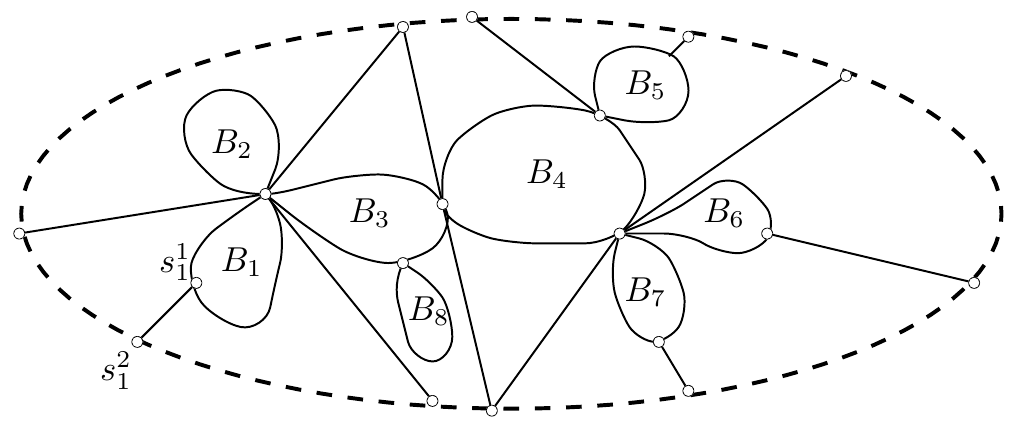}
\caption{Edges of $B$ do not cross each other in $\sigma_1$.}
\label{fig:cutvertex}
\end{figure}

\par Label each block as $B_i$ seen in the order of the clockwise walk taken on the outer boundary of $[V_1]$ (see Figure \ref{fig:cutvertex}). 
We claim that blocks of $[V_1]$ can be contracted in a certain manner to obtain a biconnected outerplanar graph. The block $B_i$ labelled last (\textit{i.e.} the block with the largest label $i$) can share at most one cut vertex with other blocks. Let $G_{i-1}$ be the graph obtained upon contracting $B_i$ into the cut vertex, (let the vertex remaining after contraction retain the label of the cut vertex). Hence, the inner layer of $G_{i-1}$ induces a maximal outerplanar graph with $i-1$ blocks, where $B_{i-1}$ is a block that shares at most one cut vertex with the remaining blocks. We continue contracting successively (producing graphs $G_k$ with $k$ blocks) till we obtain a maximal 2-outerplanar graph $G_1=G'$ with $B_1$ as the vertex set of its inner layer. Thus the inner layer induces a single biconnected component $[B_1]$. Hence, $G'$ satisfies the constraints of Lemma \ref{lemma:connected}.
Let $\sigma_2'$ be the second permutation of $V(G')$ according to Lemma \ref{lemma:connected}.  
We now obtain two permutations for $V(G)$: the first from the walk described above and the second from $\sigma_2'$.\\

\par \noindent \textit{Constructing $\sigma_1$:} We define $\sigma_1= (s_1^2,$ \textit{permutation of} $V_1, s_2^2, s_3^2, \dots, s_{n_2}^2)$ where the permutation of $V_1$ is listed according to the walk described above while labelling the vertices. It is to be noted that this walk either begins with $s_2^1$ or $s_1^1$ based on whether $e'$ has $s_1^2$ as en endpoint or $s_1^1$ as an endpoint respectively (see Lemma \ref{lemma:biconnected}).\\
\textit{Constructing $\sigma_2$:} Consider the graph $G_2$. Let $s_c^1$ be the cut vertex shared by $B_2$ and $B_1$. Suppose in $\sigma_1$, vertices of $B_2$ are seen in the relative order $s_c^1,s_b^1,s_{b+1}^1,s_{b+2}^1,\dots,s_{b+|B_2|-2}^1$, where $s_b^1$ is the first vertex of $B_2$ seen after $s_c^1$. Let $\alpha=(s_c^1,s_b^1,s_{b+1}^1,$ $s_{b+2}^1,\dots,s_{b+|B_2|-2}^1)$. We construct the second permutation by replacing $s_c^1$ in $\sigma_2'$ by $reversal(arc\mbox{-}removal(\alpha,l))$. Hence, to construct a permutation of $V(G_i)$, we replace the cut vertex in $B_i$ by a permutation of block $B_{i}$ in a similar manner. We continue inserting such permutations of blocks until we obtain a permutation of $V(G)$, which we fix as the second permutation $\sigma_2$.\\ 

\par We prove the lemma by induction on the number of blocks of $[V_1]$. If $[V_1]$ has only one block (and hence no cut vertex), it is easy to see that the permutations of $G$ are as obtained from the previous lemma. Suppose the lemma holds for all maximal 2-outerplanar graphs whose inner layer induces at most $j-1$ blocks. Let $G$ be a maximal 2-outerplanar graph whose inner layer induces exactly $j$ blocks. Let $s_1^2$ be a start vertex of $G$. We construct $\sigma_1$ as explained earlier. We find the last block $B_j$ containing only one cut vertex, say $s_c^1$, and contract $B_j$ to produce the graph $G'$. Let the vertex set of the inner layer of $G'$ be $B'$. Since the induction hypothesis can be applied to $G'$, let $\sigma^{'}=\{\sigma_1^{'},\sigma_2^{'}\}$ be the two permutations separating the edges of $G'$ as described above, with $s_1^2$ as the start vertex. Let $\sigma_2$ be produced from $\sigma_2'$ after inserting a suitable permutation of $B_j$ (refer construction of $\sigma_2$) in the second permutation. Let $\sigma=\{\sigma_1,\sigma_2\}$. We now prove that $\sigma$ separates the edges of $G$. It is easy to see that $\sigma=\sigma'$ if we ignore the vertices of $B_j \setminus \{s_c^1\}$ in $\sigma$. Hence, edges of $G$ that are separated in $\sigma'$ are separated in $\sigma$.\\ 

\par \noindent \textit{$E_1-E_1$ crossing, $E_2-E_2$ crossing, $E_1-E_2$ crossing:} It can be observed that in $\sigma_1$ there is no pairwise crossing of edges that are both taken from $E_i$, since $[V_i]$ is 1-outerplanar and thus applying Lemma \ref{lemma:outerplanar_perm}. Clearly, no edge of $E_1$ crosses any edge of $E_2$ in $\sigma_1$.\\

\noindent Note that all the vertices of $V_2$ which are adjacent to a vertex in $B_j$, are adjacent to $s^1_c$ in $G'$. Let $s_b^1$ be the first vertex of $B_j$ seen after $s_c^1$ in $\sigma_1$. In $\sigma_1$, let the vertices of $B_j$ be seen in the order $s_c^1,\dots,s^1_{b}, s^1_{b+1},\dots,s^1_{b + |B_j| - 2}$. Consider $\alpha = (s_c^1,s^1_{b}, s^1_{b+1},\dots,s^1_{b + |B_j| - 2})$. To construct $\sigma_2$ we replace $s_c^1$ in $\sigma'_2$ by $reversal(arc\mbox{-}removal(\alpha,l))$. We now discuss the rest of the crossings.\\
\par \noindent \textit{$E_1-E_{12}$ crossing:} The crossings between edges of $E_1$ and $E_{12}$ are resolved by performing the arc removal technique on $\alpha$ (see Claim \ref{claim:reversalworks}). \\
\textit{$E_{12}-E_2$ crossing:} It is not possible that an edge from $E_{12}$ crosses an edge from $E_2$ in $\sigma_2$ except for $s_1^2s_{n_2}^2$, due to the planarity of $G$. The crossing of edges with $s_1^2s_{n_2}^2$ is already resolved in $\sigma_1$.\\ 
\textit{$E_{12} - E_{12}$ crossing:} When we consider the graph $G'$, any crossing involving an edge incident on $s^1_c$ is resolved in one of the two permutations $\sigma'_1$ or $\sigma'_2$. Hence it is clear that, an edge from $E_{12}$ with one endpoint on $B'\setminus \{s^1_c\}$ and an edge from $E_{12}$ with one endpoint in $B_j$ do not cross in either $\sigma_1$ or $\sigma_2$. Also, a pair of edges from $E_{12}$ with one end point each in $B_j\setminus \{s^1_c\}$ do not cross each other (as seen similarly in Lemma \ref{lemma:biconnected}). So we only have to consider the case when (see Figure \ref{fig:cut_vertex_case}) $e \in E_{12}$ such that $e=(s^1_c,s^2_k)$ for some $s_k^2\in V_2$ and $f = (s^1_l, s^2_m)$ where $s_l^1\in B_j,s_m^1\in V_2$. If $e$ and $f$ cross each other in $\sigma_1$, then they don't cross each other in $\sigma_2$ as $s^1_c$ appears after all vertices of $B_j$. \\
\end{proof}

\begin{figure}[t]
\begin{center}
    \subfigure[Edges $e,f$ crossing in $\sigma_1$.\label{fig:cut_vertex_green_green}]{\includegraphics[width=5cm]{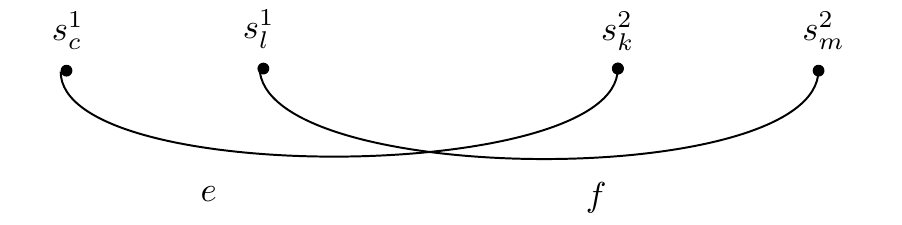}}
    \hspace{2cm}
    \subfigure[Edges $e,f$ not crossing in $\sigma_2$.\label{fig:cut_vertex_green_green_not_crossing}]{\includegraphics[width=5cm]{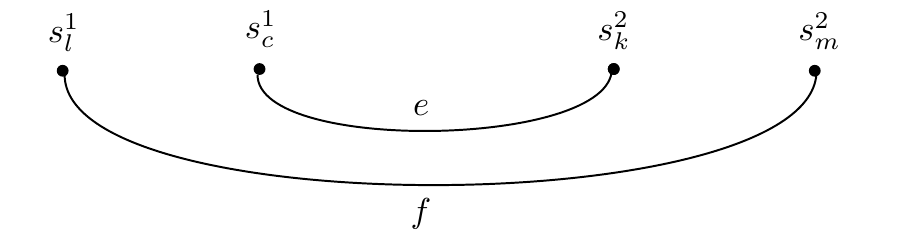}}
    \caption{$E_{12}-E_{12}$ crossing}
    \label{fig:cut_vertex_case}
\end{center}
\end{figure}

\begin{remark}
Hence, two permutations are enough to separate the edges of $G$ when $[V_1]$ contains exactly one component. It should be noted that in both permutations, the vertices of $V_2$ appear relatively according to the permutation described by the outerplanarity of $[V_2]$. It should also be noted that $s_1^2$ appears first in $\sigma_1$ and appears just after all vertices of $V_1$ in the second permutation.
\end{remark}

\begin{theorem}
The circular separation dimension of a maximal $2$-outerplanar graph is two.
\end{theorem}

\begin{proof}
\begin{figure}[t]
\centering
\includegraphics[width=0.7\textwidth]{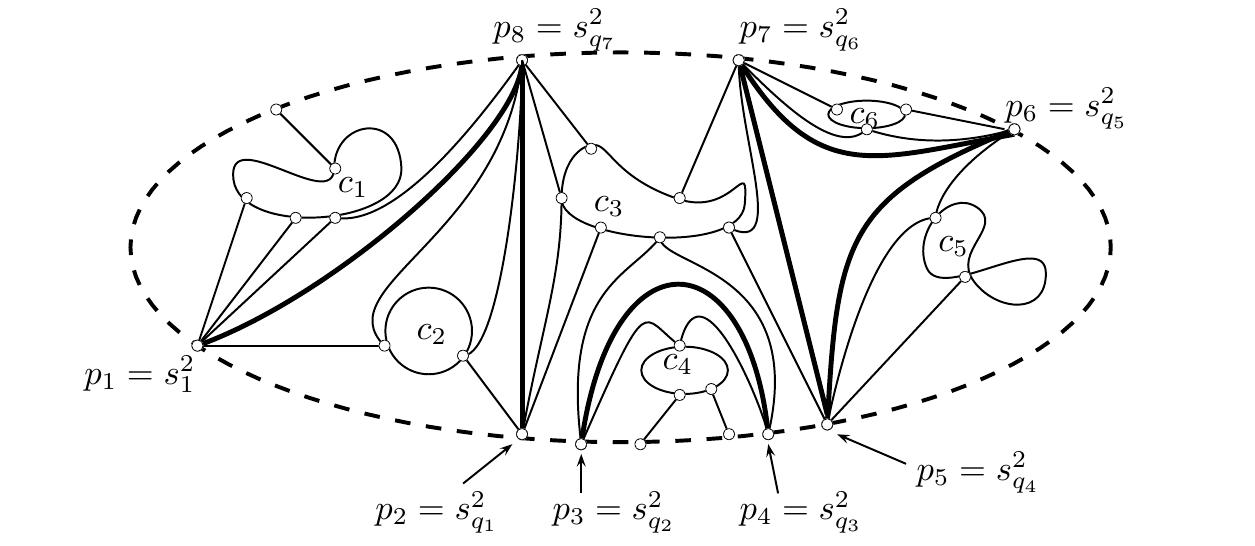}
\caption{Example of a 2-outerplanar graph where each $C^i$ is a component of $[V_1]$. Separating chords are marked in bold.}
\label{fig:components_and_chords}
\end{figure}
 As Lemma \ref{lemma:connected} deals with the case where $[V_1]$ is connected, we consider the case where $[V_1]$ is disconnected. Let each component of $[V_1]$ be $C^i$. We say that an edge $pp'\in E_2$ is a separating chord of $G$ if \textit{(i)} both $p,p'$ have a common neighbour in $V_1$ and if \textit{(ii)} removal of $p$ and $p'$ disconnects $G$ such that every component of $G-\{p,p'\}$ contains a vertex of $V_1$. It is easy to see that each $C^i$ is enclosed within the outer boundary of $[V_2]$ and by separating chords of $G$ (see Figure \ref{fig:components_and_chords}). Let $R^i$ be the smallest induced subgraph of $G$ whose vertices are those of $C^i$ and the vertices of the outer layer bounding $C^i$ including the endpoints of separating chords bounding $C^i$. We label $R^i$'s appropriately later. It is easy to see that each $R^i$ satisfies the conditions of Lemma \ref{lemma:connected}. Let $p_1$ be an endpoint of one separating chord. Let $s_1^2(=p_1),s_2^2,s_3^2,s_4^2,\dots,s_{n_2}^2$ be the ordering of vertices of $[V_2]$ by traversing its outer boundary in a counter-clockwise manner. Let $P=\{p_j\}$ be the set of endpoints of all separating chords such that $p_j$ is the $j^{th}$ endpoint while traversing the outer boundary of $[V_2]$ in a counter-clockwise manner starting at $p_1$. It is to be noted that every $p_j$ is some $s^2$. For any $R^i$ containing some $p_j$, we claim that $p_j$ can be a start vertex (recall that the start vertex has a neighbour in $V_1$). This is due to the definition of separating chords and the maximality of $G$ (as $G$ is a maximal 2-outerplanar graph).\\
\indent Let the start vertex for each $R^i$ be $p_j$ where $R^i$ contains $p_j$ and $j$ is the smallest such integer. Let $\gamma^{i}=\{\gamma_1^{i},\gamma_2^{i}$\} be the pairwise suitable family for the graph $R^i$ according to Lemma \ref{lemma:connected}. 
For $k=1,2$, let $\alpha_k^i$ be the sub-permutation of $\gamma_k^i$ restricted to the $V(C^i)$.\\

\par \noindent \textit{Construction of $\sigma$:} To form the first permutation $\sigma_1$ of $G$, we first list the start vertex $p_1=s_1^2$. We then consider all the neighbours of $p_1$ and see the edges between $p_1$ and its neighbours clockwise from the edge $p_1s_{n_2}^2$ (see Figure \ref{fig:components_and_chords}). The $R^i$'s containing $p_1$ are labelled in the order in which they are seen. We list the $\alpha_1^i$'s corresponding to the labelled $R^i$'s in increasing order of $i$. We then list the vertices of $V_2$ starting from the vertex succeeding $p_1$ (which is $s_2^2$) till we reach $p_2$. Suppose $p_2=s_{q_1}^2$. We perform the same procedure, that is we see all neighbours of $p_2$ in a clockwise manner starting from the edge $p_2s_{q_1-1}^2$ and label the $R^i$'s seen, ignoring those $R^i$ already labelled. We then list the corresponding $\alpha_1^i$'s. 
Thus for the example in Figure \ref{fig:components_and_chords}, $\sigma_1=(p_1=s_1^2,\alpha_1^1,\alpha_1^2,s_2^2,s_3^2,\dots,p_2=s_{q_1}^2,\alpha_1^3,s_{q_1+1}^2,s_{q_1+2}^2,\dots,p_3=s_{q_2}^2,\alpha_1^4,s_{q_2+1}^2,s_{q_2+2}^2,\dots,p_4=s_{q_3}^2,s_{q_3+1}^2,s_{q_3+2}^2,\dots,p_5=s_{q_4}^2,\alpha_1^5,s_{q_4+1}^2,s_{q_4+2}^2,\dots,p_6=s_{q_5}^2,$ $\alpha_1^6,s_{q_5+1}^2$ $,s_{q_5+2}^2,\dots,p_7=s_{q_6}^2,s_{q_6+1}^2,s_{q_6+2}^2,\dots,$ $p_8=s_{q_7}^2,s_{q_7+1}^2,s_{q_7+2}^2,\dots,s_{n_2}^2)$. 
In order to list the second permutation, we replace each $\alpha_1^i$ by $\alpha_2^i$ in $\sigma_1$ and shift each $p_j$ to the right of all the $\alpha_2^i$'s listed immediately after it, where each $\alpha_2^i$ contains at least one vertex adjacent to $p_j$. 
Thus for the example in Figure \ref{fig:components_and_chords}, $\sigma_2=(\alpha_2^1,\alpha_2^2,p_1=s_1^2,s_2^2,s_3^2,\dots,\alpha_2^3,p_2=s_{q_1}^2,s_{q_1+1}^2,s_{q_1+2}^2,\dots,\alpha_2^4,p_3=s_{q_2}^2,s_{q_2+1}^2,s_{q_2+2}^2,\dots,p_4=s_{q_3}^2,s_{q_3+1}^2,$ $s_{q_3+2}^2,\dots,\alpha_2^5,p_5=s_{q_4}^2,s_{q_4+1}^2,s_{q_4+2}^2,\dots,\alpha_2^6,$ $p_6=s_{q_5}^2,s_{q_5+1}^2,s_{q_5+2}^2,\dots,p_7=s_{q_6}^2,s_{q_6+1}^2,$ $s_{q_6+2}^2,$ $\dots,$ $p_8=s_{q_7}^2,s_{q_7+1}^2,s_{q_7+2}^2,\dots,s_{n_2}^2)$.\\

We now prove that these two permutations form a suitable family of permutations for $G$ by induction on the number of components of $[V_1]$. The base case of having one component is already dealt with in Lemma \ref{lemma:connected}. We assume that the induction hypothesis is true for at most $n-1$ components. Suppose $[V_1]$ has exactly $n$ components.\\
Let $P^i$ be the set of all $p_j$'s contained in $R^i$ for each $i$. It is easy to see that there exists a subgraph $R^a$ that is bounded by exactly one separating chord (see Figure~\ref{fig:Ra}). In other words, $\exists a \ni P^a=\{p_t,p_{t+1}=p_{t'}\}$ where the addition is $modulo\ |P|$ (we use $modulo\ |P|$ to accommodate the case where $t=|P|$ and $t'=1$).
We discuss the case when $p_{t'}\neq p_1$. Suppose $p_t=s_{q}^2$ and $p_{t'}=s_{q'}^2$. Let $H$ be the graph induced on the vertex set $V(C^a)\cup \{p_t=s_{q}^2,s_{q+1}^2,s_{q+2}^2,\dots s_{q'-1}^2,s_{q'}^2=p_{t'}\}$. Let $G'$ be the graph obtained from $G$ by deleting all vertices of $H$ except $p_t,p_{t'}$. 
Then $G'$ is a maximal 2-outerplanar graph with exactly $n-1$ components.

\begin{figure}[t]
\centering
\includegraphics[width=0.45\textwidth]{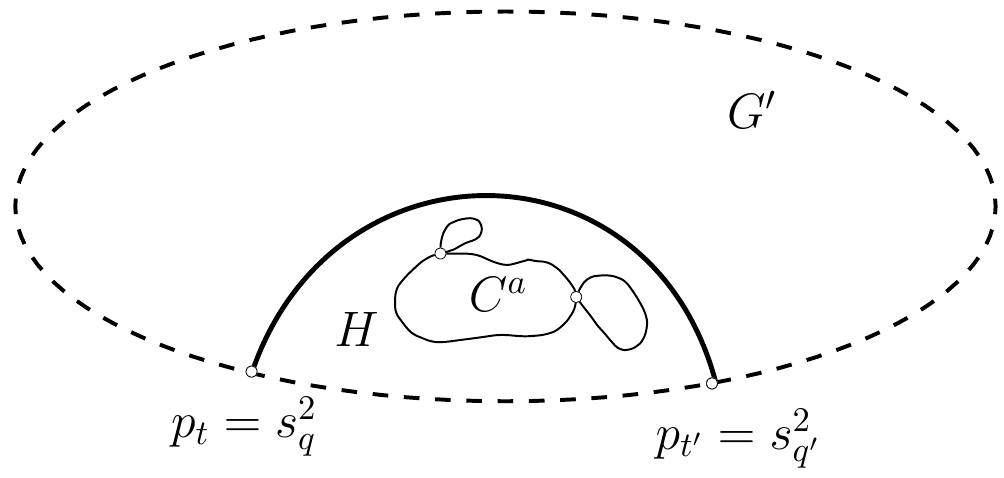}
\caption{$R^a$ with exactly one separating chord $p_tp_{t'}$.}
\label{fig:Ra}
\end{figure}

Let $\sigma'=\{\sigma_1',\sigma_2'\}$ be the pairwise suitable family for $G'$ according to the above construction with $p_1$ as the start vertex (observe that all $p_i$'s of $G$ are retained in $G'$). As defined earlier, let $\gamma^{i}=\{\gamma_1^{i},\gamma_2^{i}$\} be the pairwise suitable family for the graph $R^i$ with $p_t$ as the start vertex. For $k=1,2$, let $\alpha_k^i$ be the sub-permutation of $\gamma_k^{i}$  restricted to $V(C^i)$.\\
Let $\sigma=\{\sigma_1,\sigma_2\}$ be the pairwise suitable family for $G$ obtained from the method of construction described above, with $p_1$ as the start vertex. It is easy to see that the same permutation $\sigma_1$ can be obtained from $\sigma_1'$ by inserting $\alpha_1^a,s_{q+1}^2,s_{q+2}^2,\dots,s_{q'-1}^2$ immediately after all components following $p_t$ while $\sigma_2$ can be obtained from $\sigma_2'$ by inserting $\alpha_2^a$ immediately before $p_t$ and inserting $s_{q+1}^2,s_{q+2}^2,\dots,s_{q'-1}^2$ immediately after $p_t$. (If $p_{t'}=p_1$, then $\alpha_1^a$ appears immediately after $p_1$ in $\sigma_1$, and in $\sigma_2$, $\alpha_2^a$ appears at the very beginning. We observe that $s_{q+1}^2,s_{q+2}^2,\dots,s_{q'-1}^2$ appears immediately after $p_t$ in both $\sigma_1$ and $\sigma_2$.) It is easy to see that the relative ordering of $\gamma_k^{i}$ is maintained in $\sigma_k$ for each $i$ and $k \in \{1,2\}$. We only need to check whether an edge of $H$ can cross with an edge not contained in $H$. Let $R^b$ be a subgraph such that $b \neq a$. We now prove that $\sigma$ is a pairwise suitable family for $G$.\\

\noindent\textit{$E_1 - E_1$, $E_1 - E_2$, $E_2 - E_2$ crossings:} Clearly, an edge of $C^a$ will not cross an edge of $C^b$ since all vertices of $C^a$ appear consecutively. Therefore, a pair of edges from $E_1$ cannot cross each other. Similarly, an edge in $E_1$ cannot cross an edge in $E_2$. Clearly, a pair of edges from $E_2$ cannot cross each other as vertices of $[V_2]$ are listed according to Lemma \ref{lemma:outerplanar_perm}. \\ 
\textit{$E_1 - E_{12}$ crossings:} Suppose $e\in E(C^a)$ and $f\in E_{12} \cap E(R^b)$. Since the vertices of $C^a$ appear consecutively in both permutations, it is not possible for $f$ to cross $e$. If $e\in E(C^b)$ and $f\in E(R^a)$, the same reason holds as vertices of $C^b$ are written consecutively in both permutations.\\
\textit{$E_{2}-E_{12}$ crossings:} Suppose $e\in E_2\setminus E(H)$ and $f\in E_{12} \cap E(R^a)$. Consider only the vertices $V_2\cup V(H)$ in $\sigma_1$. Since all vertices of $H$ appear consecutively, it is not possible for $f$ to cross $e$. Similar reason holds if $e\in E_{2} \cap E(R^a)$ and $f\in E_{12}\setminus E(H)$.\\
\textit{$E_{12}-E_{12}$ crossings:} Suppose where $e\in E_{12}\cap E(R^a)$ and $f\in E_{12} \cap E(R^b)$. Note that the only vertices of $R^a$ that $f$ can be incident on are $p_t,p_{t'}$. \textit{Case (i): $R^b$ has a start vertex other than $p_t$.} Then in $\sigma_1$, all vertices of $C^b$ appear either before $p_t$ or after $p_{t'}$ and the vertices of $R^a$ appear consecutively. (if $p_{t'}=p_1$, then then all vertices of $R^b$ appear together). Hence there is no crossing between $e$ and $f$. \textit{Case (ii): $R^b$ has $p_t$ as the start vertex} (if $p_{t'}=p_1$, then suppose $R^b$ has $p_{1}$ as the start vertex). 
Consider only the vertices $V_2\cup V(R^a)\cup V(R^b)$ in $\sigma$. The respective permutations are $(p_1,...,p_t=s_q^2,\alpha_1^b,\alpha_1^a,s_{q+1}^2,s_{q+2}^2,\dots,p_{t'}s_{q'}^2,s_{q'+1}^2,s_{q'+2}^2,\dots,s_{n_2}^2)$ and $(p_1,...,\alpha_2^b,\alpha_2^a,p_t=s_q^2,s_{q+1}^2,s_{q+2}^2,\dots,p_{t'}=s_{q'}^2,s_{q'+1}^2,s_{q'+2}^2,\dots,s_{n_2}^2)$. 
In the first permutation, crossing between $e$ and $f$ can occur only if $e$ is incident on $p_t$ and $f$ is incident on $p_{t'}$ or a vertex to its right. It can be seen that this crossing is resolved in the second permutation.\\
\par Thus, a pairwise suitable family of two permutations is constructed for a maximal 2-outerplanar graph.
\end{proof}

\section{Series-Parallel graphs}
\label{sec: series}
Each series-parallel graph can be $k$-outerplanar for some natural number $k$. Interestingly, circular separation dimension of series-parallel graphs is at most 2. The proof is as follows.
\begin{theorem}
If $G$ is a series-parallel graph, then $\pi^\circ(G)\leq 2$.
\end{theorem}

\begin{proof}
We prove the statement by induction on the number of series or parallel operations. It is easy to see that the statement holds for a series or parallel operation on a single edge. It is easy to see that a parallel operation maintains the separation dimension of any graph. Hence, the theorem has to be proved only for series operations. Suppose the statement holds for $\leq k$ operations and $G'$ is a graph obtained from $k$ operations. Then, $\pi^\circ(G')\leq 2$. Let $\{\sigma_1',\sigma_2'\}$ be a pairwise suitable family for $G'$ (if $\pi^\circ(G') = 1$, then we can assume $\sigma_1' =\sigma_2'$).
Let $G$ be a graph obtained from $G'$ through a series operation. 
Let $x$ be the new vertex added where $N(x)=\{a,b\}\subseteq V(G)$. 
We obtain  $\sigma_1$ from  $\sigma_1'$ by replacing $a$ by $a,x$,  and $\sigma_2$ from  $\sigma_2'$ by replacing $b$ by $b,x$. Thus no edge can cross $ax$ in  $\sigma_1$, and $bx$ in  $\sigma_2$  as the vertices appear consecutively in $\sigma_1$ and $\sigma_2$, respectively. Hence $\{\sigma_1,\sigma_2\}$ is a pairwise suitable family for $G$.

\end{proof}

\section{Conclusion}
In this article, we show that the circular separation dimension of a 2-outerplanar graph is exactly two.  
It is to be noted that if the 2-outerplanar embedding is given as an  input, one can construct two pairwise suitable circular permutations in polynomial time.
If our conjecture on circular separation dimension of planar graphs is not true, then  an obvious question to ask  is: \emph{what is the maximum value of $k$ for which the  circular separation dimension of a maximal $k$-outerplanar graph is exactly $2$?}  We conclude with this  open question for future research.

\section*{Acknowledgements} The authors wish to thank L. Sunil Chandran for helpful discussions.

\bibliographystyle{abbrvnat}

\end{document}